\documentclass[conference]{IEEEtran}
\usepackage{graphicx, epsfig, paralist, algorithm, algorithmic}
\usepackage{amsmath}

\begin{document}

\title{Optimal power control in Cognitive MIMO systems with limited feedback}

\author{
\authorblockN{George A. Ropokis \authorrefmark{1} \authorrefmark{2},
David Gesbert \authorrefmark{1},
Kostas Berberidis \authorrefmark{1} \authorrefmark{3}}

\authorblockA{ \authorrefmark{2}
Computer Technology Institute and Press ``Diophantus", 26500 Rio-Patras, Greece
}
\authorblockA{\authorrefmark{1}Mobile Communications Dept., EURECOM, 06410, BIOT France, \\
}
\authorblockA{ \authorrefmark{3}
Dept. of Computer Engineering and Informatics, University of Patras, 26500 Rio-Patras, Greece
}
E-mail: \{ropokis,gesbert\}@eurecom.fr, berberid@ceid.upatras.gr
}

\maketitle

\newtheorem{theorem}{Theorem}

\begin{abstract}
In this paper, the problem of optimal power allocation in Cognitive Radio (CR) Multiple Input Multiple Output (MIMO) systems
is treated. The focus is on providing limited feedback solutions aiming at maximizing the secondary system rate subject
to a constraint on the average interference caused to primary communication. The limited feedback solutions are obtained
by reducing the information available at secondary transmitter (STx) for the link between STx and the secondary receiver (SRx) as
well as by limiting the level of available information at STx that corresponds to the link between the STx and the
primary receiver PRx. Monte Carlo simulation results are given that allow to quanitfy the performance achieved by the
proposed algorithms.
 
\end{abstract}
\section{Introduction}
Cognitive Radio (CR) is considered an effective approach for coping with the spectrum scarcity problem in wireless communications systems.
Among the several techniques that fall into the category  of CR, underlay CR techniques have drawn considerable attention.
The characteristic of such techniques is the fact that they allow the Secondary Transmitter (STx) to communicate with a Secondary
Receiver (SRx) in the presence of a Primary Transmitter (PTx) communicating with a Primary Receiver (PRx), provided that the average 
or peak interference  caused by STx transmission to PRx reception is below a predefined threshold. In this context, several
underlay CR techniques aim at optimally allocating the STx transmit power in a manner such that some QoS
metric, e.g. the average rate, of secondary communication is maximized subject to (s.t.) the constraint that the average or peak STx-PRx link
interference is kept below a predefined threshold. For example, in \cite{Gha07} the optimal STx power policy is presented for 
systems operating under an average STx-PRx interference constraint, while in \cite{Kan09} several policies are derived based on 
either average or peak power constraints regarding the STx transmit power and the STx-PRx link interference power.

Both these works, as well as most related works presented in the literature, are based on the so called ``Z" channel model that assumes
that there is no interference received by SRx, that corresponds to PTx transmission. More importantly, these works,  are limited to 
the study of CR Single Input Single Output (SISO) systems. Motivated by the above, in this work we propose novel power policies for 
CR Multiple Input Multiple Output (MIMO) systems taking into account the average interference caused by secondary transmission to primary 
reception. The proposed policies are characterized by the fact that they require limited feedback sent from the several network nodes 
(PRx and SRx) to STx. More specifically the proposed policies, instead of requiring exact knowledge of the CR MIMO channel matrix 
$\mathbf{H}$,  are based on the knowledge of only the eigenvalues of matrix $\mathbf{HH}^{H}$. Moreover, further feedback reduction 
schemes are derived, by introducing additional power policies that also assume only statistical CSI for the STx-PRx link.

\section{System Model}
A CR MIMO system is considered that operates in the presence of a SISO primary link. We adopt the ``Z'' channel model where 
there is no interference caused by the primary
transmission to secondary reception. In this case, the input output relation for the CR MIMO channel is given as
\begin{equation}
\mathbf{y}_{s} = \mathbf{H}\sqrt{P}\mathbf{x}_s +  \mathbf{w}_s,
\end{equation}
where $\mathbf{H}$ is the $M_R \times M_T$ channel matrix, $M_R$ being the number of receive antennas for the cognitive MIMO system and $M_T$ 
being the number of transmit antennas for the cognitive MIMO system. Rayleigh fading is assumed where the elements of matrix $\mathbf{H}$ are i.i.d.
zero mean complex Gaussian random variables with variance known to STx. Furthermore, $\mathbf{x}_{s}$ is the $M_T \times 1$ transmitted signal vector,
assumed to be $\mathbf{x}_{s} \sim {\cal CN}(\mathbf{0}_{M_T}, \mathbf{I}_{M_T})$, where $\mathcal{CN}\left(\mathbf{m},\mathbf{C}\right)$ stands for the
complex multivariate Gaussian distribution with mean $\mathbf{m}$ and covariance matrix $\mathbf{C}$. Matrix $\mathbf{I}_{M}$ stands for the $M\times M$
identity matrix and $\sim$ denotes equality in terms of distribution. In addition, $P$ is the per antenna transmit power and 
$\mathbf{w}_s$ is the Additive, White Gaussian Noise (AWGN) $M_R \times 1 $ vector.
For such a system model, the capacity of the secondary link, assuming STx has no knowledge of matrix $\mathbf{H}$, is given as \cite{Pau03}
\begin{equation}
C_s = \log_2 \det\left(\mathbf{I}_{M_R} + \frac{P}{N_0} \mathbf{H}\mathbf{H}^{H}\right)
\label{eq::mimo_cap}
\end{equation}
where $\det(\cdot)$ denotes the determinant and  $\log_{a}(\cdot)$ the logarithm with base $a$. Using standard matrix properties, \eqref{eq::mimo_cap} can be
written as
\begin{equation}
C_{s} = \sum_{i=1}^{M_R} \log_2 \left( 1 + \frac{P}{N_0} l_i \right)
\end{equation}
where $l_i$'s $i=1,\ldots,M_R$ are the non negative eigenvalues of matrix $\mathbf{H}\mathbf{H}^H$.

We assume that the above CR system operates in the presence of a primary communication link, and causes interference to primary communication. As a result, the input-output
relation for the primary communication link is expressed as
\begin{equation}
y_p = h_{pp}\sqrt{\mathcal{P}}x_p + \mathbf{h}_{sp}\sqrt{P} \mathbf{x}_{s}  + w_p
\label{eq::p_signal}
\end{equation}
where $h_{pp}$ is the PTx-PRx communication link, $\mathcal{P}$ is the constant transmit power used by the primary transmitter and $x_p \sim \mathcal{C}(0,1)$ the transmit symbol of
the primary transmitter. On the other hand,  $\mathbf{h}_{sp}$ is the $1\times M_T$ channel vector between the secondary transmitter and the primary receiver
 and $w_p$ the AWGN at the primary receiver side. By introducing the random variable $\eta = ||\mathbf{h}_{sp}||^2$,
we can then write the average interference caused by secondary transmission to primary reception as
\begin{equation}
{\mathcal U} = E\left\{\eta P \right \}.
\label{eq::interference}
\end{equation}
In the following analysis we will assume that the statistics, i.e. $\bar{\eta} = E \left\{\eta \right\}$ and the probability density function of $\eta$, are known to STx.
Having defined the secondary system capacity as well as the average interference caused by the STx-PRx link,
in the following sections we will present novel algorithms targeting on the maximization of the average achievable rate of the CR MIMO system.
This maximization takes place s.t. a constraint on the average STx-PRx interference and a constraint on the maximum STx
allowable peak power. Moreover the algorithms are designed in order to operate in the presence of limited feedback  at STx.
In the following sections, the novel algorithms are categorized according to the presumed PRx-STx feedback regarding 
the interference $\eta$ as well as the level of SRx-STx feedback concerning the matrix $\mathbf{H}$.

\section{Algorithms based on instantaneous STx-PRx knowledge}
As a starting point, we consider that STx has exact instantaneous knowledge concerning the STx-PRx interference channel. Moreover, adopting
a limited SRx-STx feedback scenario, we assume that STx has instantaneous knowledge of only the eigenvalues $l_i$'s $i=1\ldots,M_{R}$ of matrix $\mathbf{HH}^{H}$.
To the best of our knowledge this hypothesis has never been tested in the context of CR-MIMO links. The knowledge of such limited information does not 
suffice for applying the  well known singular value decomposition (SVD) based approach \cite{Pau03} and the use of the waterfilling algorithm for
capacity maximization. Nevertheless one can exploit the knowledge of eigenvalues in order to optimally allocate power across channel states. To this
end, STx can apply power allocation across channel states by solving the following optimization problem
\begin{equation}
\begin{split}
{}&{}\textrm{maximize :}\;\;\;
C_{s} = E \left\{ \sum_{i=1}^{M_R} \log_2 \left( 1 + \frac{P}{N_0} l_i \right) \right\} \\&
\textrm{s.t.}\;\;\;{\mathcal U} = E\left\{\eta P \right \} \leq Q, \;\; 0 \leq P \leq P_{\max}
\end{split}
\end{equation}
or equivalently, after a change of the logarithm's base, as
\begin{equation}
\begin{split}
{}&{}\textrm{maximize :}\;\;\;
E \left\{ \sum_{i=1}^{M_R} \ln \left( 1 + \frac{P}{N_0} l_i \right) \right\} \\&
\textrm{s.t.}\;\;\;{\mathcal U} = E\left\{\eta P \right \} \leq Q, \;\; 0 \leq P \leq P_{max}.
\end{split}
\label{eq::opt1}
\end{equation}
Notice that in this optimization problem, expectation is taken over the channel eigenvalues and
the random variable $\eta$. Therefore in addition to knowledge of $\mathbf{H}$ and $\eta$, also knowledge
of the second order statistics, i.e., the variance of the elements of the $\mathbf{H}$ channel matrix,
and the statistics of $\eta$, i.e. its distribution, is required. Using this information, the power policy
varies as a function of $\mathbf{H}$ and $\eta$ while taking into account the average interference,
positivity and peak power constraints. In what follows we will refer to the policy that solves the above optimization problem as the 
Eigenvalue Based Power Policy (EBPP). EBPP is presented in the following subsection.
\subsection{Derivation of the EBPP}
The EBPP  can be found by first checking if the inequality 
\begin{equation}
\bar{\eta} P_{max}\leq Q
\end{equation}
is satisfied. If this holds, then, the optimal power policy EBPP is readily expressed as
\begin{equation}
P_{EBPP} = P_{\max}.
\end{equation}
Otherwise, the optimization problem presented in \eqref{eq::opt1} has the solution presented in the following theorem
\begin{theorem}
If $\bar{\eta} P_{max} > Q$ the optimal power policy for optimization problem in \eqref{eq::opt1} is given by
\label{theor::1}
\small
\begin{equation}
P_{EBPP} = \begin{cases}
   0 &  \eta \geq \frac{\sum_{i=1}^{M_R}l_i}{\lambda N_0} \\
	 P_{max} &  \eta \leq  \sum_{i=1}^{M_R}\frac{ {l_i}}{\lambda\left(N_0 + P_{max} l_i\right)}\\
   \rho\left(\sum_{i=1}^{M_R}\frac{l_i}{N_0 + P l_i} - \lambda \eta =0 \right) & \text{otherwise} 
  \end{cases}
\end{equation}
\normalsize
where $\rho(f(P))$ stands for the root of equation $f(P)=0$ with respect to  $P$, and $\lambda$ is selected such that $ \bar{\eta}P_{EBPP} = Q$.
\end{theorem}
\begin{proof}
The optimization problem in \eqref{eq::opt1} can be solved by noticing that due to concavity of the objective function and convexity of constraint functions,
it is convex. Thus, the solution to this optimization problem can
be found by applying KKT conditions which are expressed as follows
\begin{equation}
\begin{split}
{}&{}\sum_{i=1}^{M_R} \frac{l_i}{N_0 + P l_i} - \lambda \eta + \mu -\nu = 0 \\&
\mu P = 0, \mu \geq 0, P \geq 0, \\&
\nu (P-P_{max}) = 0 , \nu \geq 0, P-P_{max} \leq 0 \\&
\lambda (E\left\{ \eta P \right\} - Q) = 0, \lambda \geq 0, E\left\{ \eta P \right\} \leq Q
\end{split}
\end{equation}
where $\lambda$ is a Lagrange Multiplier corresponding to the average interference constraint and $\mu, \nu$ are multipliers corresponding to the
non negativity $P \geq 0$ and peak power $P \leq P_{max}$ constraints respectively.
Based on KKT conditions, the following cases need to be examined separately.
\begin{itemize}
\item $P_{EBPP}=0$: In this case, it holds that $\mu \geq 0$ and that  multiplier $\nu$ equals $0$
\footnote{This is due to the fact that the corresponding constraint, i.e. $P\leq P_{max}$ is inactive}.
KKT conditions then lead to the inequality
\begin{equation}
\eta \geq \frac{\sum_{i=1}^{M_R}l_i}{\lambda N_0}
\end{equation}
\item $P_{EBPP} = P_{max}$: In this case is holds that $\nu \geq 0$ and that $\mu =0$ since constraint $P=0$ is inactive.
Therefore KKT conditions lead to the following inequality for this case
\begin{equation}
\eta \leq \sum_{i=1}^{M_R} \frac{l_i}{\lambda(N_0 + l_i P_{max})}.
\end{equation}
\item $P_{EBPP} > 0$: In this case, since constraints $0 \leq P \leq P_{max}$ are inactive, it holds that $\mu = \nu = 0$ As a result, KKT conditions state that the optimal policy  $P$ should satisfy the
equation
\begin{equation}
\sum_{i=1}^{M_R} \frac{l_i}{N_0 +P_{EBPP}l_i} - \lambda \eta = 0.
\end{equation}
where the solution of this equation can be found using an iterative root finding  algorithm.
\end{itemize}
Moreover, by KKT conditions, it follows that the Lagrange
Multiplier $\lambda$ must be chosen such that $E\left\{\eta P\right\} = Q$.
\end{proof}

Having derived the solution to the power alocation problem \eqref{eq::opt1}, in the following subsection we present a second power policy, namely the
Maximum Eigenvalue Based Power Policy (MEBPP) that further reduces the need for feedback on the SRx-STx link since it requires only the knowledge
of the maximum eigenvalue of $\mathbf{HH}^{H}$. 
\subsection{Derivation of MEBPP}
In this section we derive a second power policy that is based on the knowledge of only the maximum eigenvalue $l_{max}$ of matrix $\mathbf{H}\mathbf{H}^{H}$.
Based on this knowledge, STx can formulate the following optimization problem
\begin{equation}
\begin{split}
{}&{}\textrm{maximize:}\;\;\;E\left\{ \ln \left(1 + \frac{P}{N_0} l_{\max}\right )\right\} \\&
\textrm{s.t.}\;\;\;  E\left\{\eta P\right\} \leq Q, 0 \leq P \leq P_{max}.
\end{split}
\label{eq::opt2}
\end{equation}
In case that $\bar{\eta} P_{max} \leq Q$ the solution to this optimization problem is simply $P_{MEBPP} = P_{\max}$.
Otherwise, the solution to this problem is stated in the following theorem:
\begin{theorem}
\label{eq::theor2}
The solution to optimization problem \eqref{eq::opt2} if $\bar{\eta} P_{max} > Q$ is given as 
\begin{equation}
P_{MEBPP} =  \min\left\{ \left[ \frac{1}{\lambda \eta} - \frac{N_0}{l_{max}} \right]^{+}, P_{max} \right \}
\end{equation}
where $\lambda$ is the Lagrange multiplier, selected in order to satisfy the constraint $E\left\{ \eta P_{MEBPP} \right\} = Q$.
\end{theorem}
\begin{proof}
The proof is derived following the steps of proof of Theorem \ref{theor::1} while setting $M_R=1$ and substituting $l_1$ by $l_{max}$
\end{proof}
Having presented the two power policies, in the following section we will present further limited feedback
power policies. 
\section{ Power policies for further limited channel feedback}
In this section  we present two new power policies that are based on a further limited feedback scenario.
In more detail we examine the case that the secondary transmitter has only knowledge of the statistics  
of $\eta= ||\mathbf{h}_{sp}||^2$, i.e.  knowledge of $\bar{\eta}$ and the distribution of $\eta$. In this case, the constraint on the average interference
caused by STx to PRx is given as 
\begin{equation}
{\mathcal V} = E\left\{ \eta P \right\} = \bar{\eta} E\left\{ P \right\}.
\label{V_constraint}
\end{equation}
Using ${\mathcal V}$, in the
following two subsections, we derive two new limited feedback power
policies, the Interference statistics and eigenvalue based power policy (IEBPP) and the
Interference statistics and maximum eigenvalue based  power policy (IMEBPP).

\subsection{Derivation of IEBPP}
The policy IEBPP is essentially the solution to the following optimization problem
\begin{equation}
\begin{split}
{}&{}\textrm{maximize :}\;\;\;
E \left\{ \sum_{i=1}^{M_R} \ln \left( 1 + \frac{P}{N_0} l_i \right) \right\} \\&
\textrm{s.t.}\;\;\;{\mathcal V} = \bar{\eta}  E\left\{ P \right \} \leq Q, \;\; 0 \leq P \leq P_{max}.
\end{split}
\label{eq::opt3}
\end{equation}
In case that $\bar{\eta}P_{max} \leq Q$ the solution to this problem is the trivial fixed policy $P = P_{\max}$.
Otherwise, for the case $\bar{\eta} P_{max} > Q$, the derivation of the IEBPP is also based on the use of Lagrange Multipliers in a way
similar to the one presented in Theorem 1. As a result, IEBPP is defined as $P_{IEBP} = P_{max}$ when $\bar{\eta} P_{max}  \leq Q$.
Otherwise, $P_{IEBPP}$ is defined as shown in the following theorem.
\begin{theorem}
In case $\bar{\eta} P_{max} > Q$  The power policy  is expressed as:
\small
\begin{equation}
P = \begin{cases}
   0 & \bar{\eta} \geq \frac{\sum_{i=1}^{M_R}l_i}{\lambda N_0} \\
	 P_{max} & \bar{\eta} \leq \sum_{i=1}^{M_R} \frac{ {l_i}}{\lambda\left(N_0 + l_i P_{max}\right)}\\
   \rho\left(\sum_{i=1}^{M_R}\frac{l_i}{N_0 + P l_i} - \lambda \bar{\eta} =0 \right) & \text{otherwise} 
  \end{cases}
\end{equation}
\normalsize
\label{theor::3}
where $\lambda$ is the lagrange multiplier that is selected such as to satisfy the constraint $\bar{\eta}E\left\{P_{IEBPP}\right\} = Q$
\end{theorem}
\begin{proof}
The above theorem can be proved by introducing the Lagrangian function and using KKT conditions. 
It is easy to see then that the problem is similar to optimization problem \eqref{eq::opt1}
with $\eta$ being substituted by $\bar{\eta}$. The proof of Theorem \ref{theor::3} is then obtained
following the procedure of the proof \ref{theor::1} and substituting $\eta$ by $\bar{\eta}$.
\end{proof}
Similar to the case of EBPP, a limited feedback version of IEBPP, namely Interference statistics and Maximum Eigenvalue
Based Power Policy (IMEBPP) can be constructed as shown in the following subsection.
\subsection{Derivation of IMEBPP}
The IMEBPP is derived as the solution to the following optimization problem
\begin{equation}
\begin{split}
{}&{} \textrm{maximize}\;\; E \left\{ \ln\left( 1 + \frac{P}{N_0} l_{\max} \right) \right\} \\&
\textrm{s.t.}\;\; \bar{\eta}E\left\{ P \right\} \leq Q,\;\; 0 \leq P \leq P_{max} 
\end{split}
\label{eq::opt4}
\end{equation}
where it is assumed that STx has knowledge of only the maximum eigenvalue $l_{\max}$ of matrix $\mathbf{H}$
and the statistics of the elements of matrix $\mathbf{H}$ mentioned earlier as well as knowledge of the statistics of $\eta$. As in the previous
cases, the policy is derived by examining the two following cases.
\begin{inparaenum}
\item $ \bar{\eta}P_{max} \leq Q$: In this case the optimal policy is the policy $P_{IMEBPP}=P_{\max}$
\item $\bar{\eta}P_{max} >Q$ In this case the optimal policy is described in the following theorem.
\end{inparaenum}
\begin{theorem}
\label{theor::4}
The solution of problem \eqref{eq::opt4}, in case that $\bar{\eta}P_{max} > Q$ leads to the following power allocation scheme
\begin{equation}
P = \min \left\{ \left[ \frac{1}{\lambda \bar{\eta}} - \frac{N_0}{l_{\max}}\right]^{+}, P_{\max} \right\}
\end{equation}
where $\lambda$ is the Lagrange multiplier, selected in order to satisfy the constraint $\bar{\eta}E\left\{ P \right\} = Q$.
\end{theorem}
\begin{proof}
The solution is found following the procedure for the proof of Theorem \eqref{theor::3} and setting $M_R=1$ and $l_1 = l_{max}$.
\end{proof}

\section{Comparison of the derived policies}
In this section we compare the several policies in terms of achievable rate. Specifically, in Fig. \ref{fig::2_times_2}
we present the achievable rate for the several policies for a $2\times 2$ MIMO system with $\bar{\eta} = M_T$ and
$\mathbf{H}$ being a complex matrix with i.i.d entries being zero mean Gaussian and unit variance. 
\begin{figure}
\centering
\includegraphics[width=3.0in]{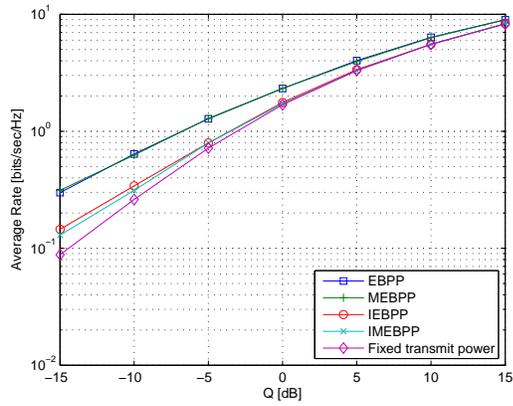}
\caption{Achievable performance of the presented policies in terms of average rate for a $2 \times 2$ system}
\label{fig::2_times_2}
\end{figure}
We can observe that the policies EBPP and MEBPP achieve very similar results. Moreover, policies IEBPP and IMEBPP achieve also very
similar results. Therefore, it can be deduced that knowledge of only the maximum eigenvalue of the channel matrix can lead to very
similar performance as compared to the case that all eigenvalues are known. On the other hand we can see that there is a notable
difference between the performance of Policies EBPP and MEBPP as compared to the performance of IEBPP and IMEBPP. This difference is caused by
the fact that in policies IEBPP and IMEBPP only statistical knowledge of interference channel $\eta$ is assumed. Finally, in our plot,
we also include a curve corresponding to a fixed transmit policy  with transmit power equal to $P = Q / \bar{\eta}$. Comparing
this curve with curves corresponding to policies EBPP and MEBPP we can observe that even the knowledge of eigenvalue $l_{\max}$ and 
the exact knowledge of $\eta$ can lead to a notable performance gain in terms of achievable rate. Moreover, in Fig. \ref{fig::5_times_5} the performance of a $5 \times 5$
MIMO system is ploted for the several policies that we have developed along with the performance of the fixed power policy with
transmitted power $P = Q/ \bar{\eta}$. As it can be seen from this figure, the performance gap between the several policies
decreases and Policies IEBPP, IMEBPP almost coincide, in terms of achievable rate with the fixed transmit policy. On the other hand,
although the performance gap between Policies EBPP, MEBPP and the fixed transmit power decreases, there exists a clear performance
gain for policies EBPP and MEBPP against fixed transmit power and policies IEBPP and IMEBPP. Finally, by comparing results plotted in
Figs. \ref{fig::2_times_2} and \ref{fig::5_times_5} it is evident that for all applied policies, the achievable gain
increases as the number of transmit/receive antennas increases.
\begin{figure}
\centering
\includegraphics[width=3.0in]{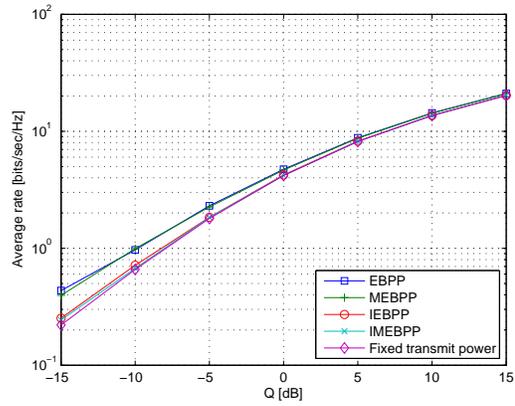}
\caption{Achievable performance of the presented policies in terms of average rate for a $5\times 5$ system.}
\label{fig::5_times_5}
\end{figure}
\section*{Acknowledgement}
This paper is part of the MIMOCORD project. The MIMOCORD research project is implemented within the framework of the Action ``Supporting Postdoctoral Researchers'' of the Operational Program 
``Education and Lifelong Learning'' (Action’s Beneficiary: General Secretariat for Research and Technology), and is co-financed by the European Social Fund (ESF) and the Greek State.

\bibliographystyle{IEEEtran}
\bibliography{IEEEabrv,refs}
\end{document}